\documentclass[11pt,reqno]{amsart}
\usepackage[nosumlimits]{mathtools}
\usepackage{amsthm}
\usepackage{float}
\usepackage{graphicx, enumerate, url}
\usepackage[margin=1in]{geometry}
\usepackage{amssymb,cool}
\usepackage{algorithm}
\usepackage{algpseudocode}
\usepackage{color}
\usepackage{tikz}
\usepackage{bbm}
\usetikzlibrary{3d,calc}
\usepackage[low-sup]{subdepth}
\usepackage[normalem]{ulem}

\numberwithin{equation}{section}
\numberwithin{algorithm}{section}

\theoremstyle{plain}
\newtheorem{theorem}{Theorem}[section]
\newtheorem{proposition}[theorem]{Proposition}
\newtheorem{lemma}[theorem]{Lemma}
\newtheorem{corollary}[theorem]{Corollary}

\theoremstyle{definition}
\newtheorem{definition}[theorem]{Definition}

\theoremstyle{remark}

\newcommand{\tp}{{\scriptscriptstyle\mathsf{T}}}

\let\O\undefined

\DeclareMathOperator{\O}{O}

\DeclareMathOperator{\SO}{SO}
\DeclareMathOperator{\tr}{tr}
\DeclareMathOperator{\diag}{diag}
\DeclareMathOperator{\rank}{rank}

\DeclareMathOperator{\sinc}{sinc}
\DeclareMathOperator{\spn}{span}

\DeclareMathOperator{\GL}{GL}

\DeclareMathOperator{\V}{V}

\DeclareMathOperator{\im}{im}
\DeclareMathOperator{\Gr}{Gr}
\DeclareMathOperator{\Graff}{Graff}

\DeclareMathOperator{\T}{\mathsf{T}}

\begin{document}
\title{Numerical Algorithms on the Affine Grassmannian}
\author[L.-H.~Lim]{Lek-Heng~Lim}
\address{Computational and Applied Mathematics Initiative, Department of Statistics,
University of Chicago, Chicago, IL 60637-1514.}
\email{lekheng@galton.uchicago.edu}
\author[K.~S.-W.~Wong]{Ken~Sze-Wai~Wong}
\address{Department of Statistics,
University of Chicago, Chicago, IL 60637-1514.}
\email{kenwong@uchicago.edu}
\author[K.~Ye]{Ke Ye}
\address{KLMM, Academy of Mathematics and Systems Science, Chinese Academy of Sciences, Beijing 100190, China}
\email{keyk@amss.ac.cn}
\begin{abstract}
The affine Grassmannian is a noncompact smooth manifold that parameterizes all affine subspaces of a fixed dimension. It is a natural generalization of Euclidean space, points being zero-dimensional affine subspaces. We will realize the affine Grassmannian as a matrix manifold and extend Riemannian optimization algorithms including steepest descent, Newton method, and conjugate gradient, to real-valued functions on the affine Grassmannian. Like their counterparts for the Grassmannian, these algorithms are in the style of Edelman--Arias--Smith --- they rely only on standard numerical linear algebra and are readily computable.
\end{abstract}



\maketitle
\section{Introduction}

A $k$-dimensional affine subspace of $\mathbb{R}^n$, denoted $\mathbf{A} + b$, is a $k$-dimensional linear subspace $\mathbf{A} \subseteq \mathbb{R}^n$ translated by a displacement vector $b \in \mathbb{R}^n$. The set of all $k$-dimensional affine subspaces in $\mathbb{R}^n$ constitutes a smooth manifold called the \emph{affine Grassmannian}, denoted $\Graff(k,n)$, an analogue of the usual Grassmannian $\Gr(k, n)$ that parameterizes $k$-dimensional linear subspaces in $\mathbb{R}^n$.

The affine Grassmannian is a relatively obscure object compared to its ubiquitous cousin, the Grassmannian. Nevertheless, it is $\Graff(k,n)$, which like $\mathbb{R}^n$ is a non-compact manifold, that is the natural generalization of Euclidean space --- points are zero-dimensional affine subspaces and so $\Graff(0,n) = \mathbb{R}^n$. The non-compactness makes $\Graff(k,n)$ harder to study than $\Gr(k,n)$, which is compact. The two main objectives of our article are to (i) develop the concrete foundations for Edelman--Arias--Smith-style \cite{EAS} optimization algorithms on the affine Grassmannian; (ii) explicitly describe three such algorithms: steepest decent, conjugate gradient, and Newton method.

The  aforementioned ``Edelman--Arias--Smith-style'' deserves special  elaboration. By this, we mean that we do not view our manifold in an abstract fashion comprising charts glued together; instead we emphasize the use of global coordinates in the form of matrices for efficient computations. The affine Grassmannian then becomes  a concrete computational platform (like $\mathbb{R}^n$) on which geodesics, exponential maps, parallel transports, Riemannian gradient and Hessian, etc, may all be efficiently computed using standard numerical linear algebra.

In fact, a main reason for the widespread applicability of the Grassmannian is the existence of  several excellent choices of global matrix coordinates, allowing subspaces to be represented as matrices and thereby  the use of a vast range of algorithms in numerical linear algebra \cite{AMSbook, AMS, AMSV, EAS}. Such concrete realizations of an abstract manifold is essential for application purposes. By providing a corresponding set of tools for the affine Grassmannian, we effectively extend the wide range of data analytic techniques that uses the Grassmannian as a model for linear subspaces \cite{ER, HRS, HamLee2008, HHH, LZ, MYDF, SLLM, TVF, VidMaSas2005} to affine subspaces.

Before this work, the affine Grassmannian, as used in the sense\footnote{We would like to caution the reader that the term `affine Grassmannian'' is now used far more commonly to refer to another very different object; see \cite{AR, FG, L}. In this article, it will be used exclusively in the sense of Definition~\ref{def:graff}. If desired, `Grassmannian of affine subspaces' may be used to avoid ambiguity.}  of this article, i.e., the manifold that parameterizes $k$-dimensional affine subspaces in $\mathbb{R}^n$, has received scant attention in both pure and applied mathematics. To the best of our knowledge, our work is the first to study it systematically. We summarize our contributions in following: 
\begin{itemize}
\item In Section~\ref{sec:aff}, we show that the affine Grassmannian is a Riemannian manifold that can be embedded as an open submanifold of the Grassmannian. We introduce some basic systems of global coordinates: affine coordinates, orthogonal affine coordinates, and projective affine coordinates. These simple coordinate systems are convenient in proofs but are inadequate when it comes to actual computations.

\item In Section~\ref{sec:global}, we introduce two more sophisticated systems of coordinates that will be critical to our optimization algorithms --- Stieffel coordinates and projection coordinates --- representing points on the affine Grassmannian as $(n+1) \times (k+1)$ matrices with orthonormal columns and as $(n+1) \times (n+1)$ orthogonal projection matrices respectively. We establish a result that allows us to switch between these two systems of coordinates.

\item In Section~\ref{sec:objects}, we describe the common differential geometric objects essential in our optimization algorithms --- tangent spaces, exponential maps, geodesics, parallel transports, gradients, Hessians --- concretely in terms of Stiefel  coordinates and projection coordinates. In particular, we will see that once expressed as matrices in either coordinate  system, these objects become readily computable via standard numerical linear algebra.

\item In Section~\ref{sec:optim}, we describe (in pseudocodes) steepest descent, Newton method, and conjugate gradient in Stiefel coordinates and the former two in projection coordinates. 

\item In Section~\ref{sec:numerical}, we report the results of our numerical experiments on two test problems: (a) a nonlinear nonconvex optimization problem that arises from a coupling of a symmetric eigenvalue problem with a quadratic fractional programming problem, and (b) the problem of computing the Fr\'echet/Karcher mean of two affine subspaces. These problems are judiciously chosen --- they are nontrivial and yet their exact solutions may be determined in closed form, which in turn allows us to ascertain whether our algorithms indeed converge to their actual global optimizers. In the extensive tests we carried out on both problems, the iterates generated by our algorithms converge to the true solutions in every instance.
\end{itemize}

\section{Affine Grassmannian}\label{sec:aff}

The affine Grassmannian was first described in \cite{KR} but has received relatively little attention compared to the Grassmannian of linear subspaces $\Gr(k,n)$. Aside from a brief discussion in \cite[Section~9.1.3]{Nicolaescu}, we are unaware of any systematic treatment. Nevertheless, given that it naturally parameterizes all $k$-dimensional affine subspaces in $\mathbb{R}^n$, it is evidently an important object that could rival the usual Grassmannian in practical applicability. To distinguish it from a different but identically-named object,\footnotemark[1] we may also refer to it as the \emph{Grassmannian of affine subspaces}.

We will establish basic properties of the affine Grassmannian with a view towards Edelman--Arias--Smith-type optimization algorithms. These results are neither difficult nor surprising, certainly routine to the experts, but  have not appeared before elsewhere to the best of our knowledge.

We remind the reader of some basic terminologies. A  \emph{$k$-plane}  is a $k$-dimensional linear subspace and  a \emph{$k$-flat} is a $k$-dimensional affine subspace.  A \emph{$k$-frame} is an ordered  basis of a $k$-plane and we will regard it as an $n \times k$ matrix whose columns $a_1,\dots, a_k$ are the basis vectors. A \emph{flag} is a strictly increasing sequence of nested linear subspaces, $\mathbf{X}_0\subset \mathbf{X}_{1} \subset \mathbf{X}_2\subset \cdots $. A flag is said to be \emph{complete} if $\dim \mathbf{X}_k= k$, \emph{finite} if $k =0,1,\dots,n$, and \emph{infinite} if $k\in \mathbb{N} \cup \{0\}$. We write $\Gr(k,n)$ for the \emph{Grassmannian} of $k$-planes in $\mathbb{R}^n$, $\V(k,n)$ for the \emph{Stiefel manifold} of orthonormal $k$-frames, and $\O(n) \coloneqq \V(n,n)$ for the \emph{orthogonal group}. We may regard $\V(k,n)$ as a homogeneous space,
\begin{equation}\label{eq:stief}
\V(k,n) \cong \O(n)/\O(n-k),
\end{equation}
or more concretely as the set of $n \times k$ matrices with orthonormal columns.  
There is a right action of the orthogonal group $\O(k)$ on $\V(k,n)$: For $Q\in \O(k)$ and $A \in \V(k,n)$, the action  yields $AQ \in \V(k,n)$ and the resulting homogeneous space is $\Gr(k,n)$, i.e.,
\begin{equation}\label{eq:grass}
\Gr(k,n) \cong \V(k,n)/\O(k) \cong \O(n)/\bigl(\O(n-k) \times \O(k)\bigr).
\end{equation}
By \eqref{eq:grass}, $\mathbf{A} \in \Gr(k,n)$ may be identified with the equivalence class of its orthonormal $k$-frames $\{ AQ \in \V(k,n): Q \in \O(k)\}$. Note $\spn(AQ) = \spn(A)$ for $Q \in \O(k)$. 
\begin{definition}[Affine Grassmannian]\label{def:graff}
Let $k<n$ be positive integers. The \emph{Grassmannian of $k$-dimensional affine subspaces} in $\mathbb{R}^n$ or Grassmannian of $k$-flats in $\mathbb{R}^n$, denoted by $\Graff(k,n)$, is the set of all $k$-dimensional  affine subspaces of $\mathbb{R}^n$. For an abstract vector space $V$, we write $\Graff_k(V)$ for the set of $k$-flats in $V$.
\end{definition}
This set-theoretic definition  reveals little about the rich geometry behind $\Graff(k,n)$, which we will see  is a smooth Riemannian manifold intimately related to the Grassmannian $\Gr(k+1,n+1)$.

Throughout this article, a boldfaced letter $\mathbf{A}$ will always denote a subspace and the corresponding normal typeface letter $A$ will then denote a matrix of basis vectors (often but not necessarily orthonormal) of $\mathbf{A}$.
We denote a $k$-dimensional affine subspace as $\mathbf{A}+b \in \Graff(k,n)$ where $\mathbf{A} \in \Gr(k,n)$ is a $k$-dimensional linear subspace and $b \in \mathbb{R}^n$ is the displacement of $\mathbf{A}$ from the origin. If $A = [a_1,\dots, a_k] \in \mathbb{R}^{n \times k}$ is a basis of $\mathbf{A}$, then
\begin{equation}\label{eq:affine}
\mathbf{A}+b \coloneqq \{\lambda_1 a_1 + \dots + \lambda_k a_k + b \in \mathbb{R}^n : \lambda_1,\dots,\lambda_k \in \mathbb{R} \}.
\end{equation}
The notation $\mathbf{A} + b$ may be taken to mean a coset of the subgroup $\mathbf{A}$ in the additive group $\mathbb{R}^n$ or the Minkowski sum of the sets $\mathbf{A}$ and $\{b\}$ in the Euclidean space $\mathbb{R}^n$. The dimension of  $\mathbf{A} + b$ is defined to be the dimension of the vector space $\mathbf{A}$.  As one would expect of a coset representative, the displacement vector $b$ is not unique: For any $a \in \mathbf{A}$, we have $\mathbf{A} + b = \mathbf{A} + (a+ b)$.

Since a $k$-dimensional affine subspace of $\mathbb{R}^n$ may be described by a $k$-dimensional subspace of $\mathbb{R}^n$ and a displacement vector in $\mathbb{R}^n$, it might be tempting to guess that $\Graff(k,n)$ is identical to $\Gr(k,n) \times \mathbb{R}^n$. However, as we have seen, the representation of an affine subspace as $\mathbf{A} + b$ is not unique and we emphasize that
\[
\Graff(k,n) \ne \Gr(k,n) \times \mathbb{R}^n.
\]
Although $\Graff(k,n)$ can be regarded as a quotient of $\Gr(k,n) \times \mathbb{R}^n$, this description is neither necessary nor helpful for our purpose  and we will not pursue this point of view in our article.

We may choose an orthonormal basis for $\mathbf{A}$ so that $A \in \V(k,n)$ and choose $b$ to be orthogonal to $\mathbf{A}$ so that $A^\tp  b = 0$.  Hence we may always represent $\mathbf{A}+b \in \Graff(k,n)$ by a matrix $[A, b_0] \in \mathbb{R}^{n \times (k+1)}$ where $A^\tp  A = I$ and $A^\tp  b_0 = 0$; in this case we call $[A,b_0]$ \emph{orthogonal affine coordinates}. A moment's thought would reveal that any two orthogonal affine coordinates $[A, b_0], [A',b_0'] \in \mathbb{R}^{n \times (k+1)}$  of the same affine subspace $\mathbf{A} + b$ must have $A' = AQ$ for some $Q \in \O(k)$ and $b_0' = b_0$.

We will not insist on using orthogonal affine coordinates at all times as they can be unnecessarily restrictive, especially in proofs. Without these orthogonality conditions, a matrix $[A,b_0] \in \mathbb{R}^{n \times (k+1)}$ that represents an affine subspace $\mathbf{A} + b$ in the sense of \eqref{eq:affine} is called its \emph{affine coordinates}.

Our main goal is to show that the vast array of optimization techniques \cite{AMSbook, AMS, AMSV, EAS, newgrass} may be adapted to the affine Grassmannian. 
In this regard, it is the following view of $\Graff(k,n)$ as an embedded open submanifold of $\Gr(k+1,n+1)$  that will prove most useful. Our construction of this embedding is illustrated in Figure~\ref{fig:plane} and formally stated in Theorem~\ref{thm:alg}.
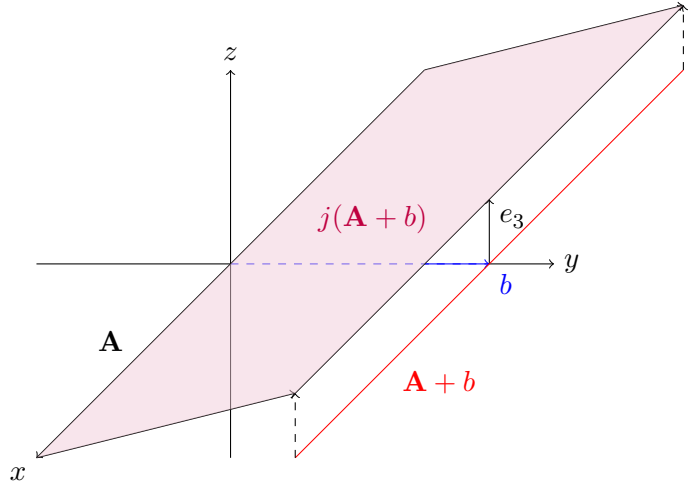
\begin{figure}[!ht]
\centering
\begin{tikzpicture}[scale=0.86]
\draw[black] (-3,0)--(0,0);
\draw[dashed,blue] (0,0)--(4,0);
\draw[->] (4,0) -- (5,0)node[anchor=west] {$y$};
\draw[->] (0,-3) -- (0,3)node[anchor=south] {$z$};
\node [above left,black] at (-1.5,-1.5) {$\mathbf{A}$}; 
\draw[black][->] (3,3) -- (-3,-3);
\node [below left,black] at (-3,-3) {$x$};
\draw[blue][->] (3,0)--(4,0);
\node [below right,blue] at (4,0) {$b$}; 
\draw[black][->]  (4,0)-- (4,1);
\node [below right,black] at (4,1) {$e_3$}; 
\draw[red] (1,-3) -- (7,3);
\node [below right,red] at (2.5,-1.5) {$\mathbf{A}+b$}; 
\draw[dashed][->]  (1,-3)-- (1,-2);
\draw[dashed][->]  (7,3)-- (7,4);
\draw[black][->] (1,-2) -- (7,4);
\draw[black][->] (-3,-3)--(1,-2);
\draw[black][->] (3,3)-- (7,4);
\fill[purple!20,opacity=0.5] (-3,-3) -- (1,-2) -- (7,4) -- (3,3) -- cycle;
\node[purple] at (2.2,0.7) {$j(\mathbf{A}+b)$};
\end{tikzpicture} 
\caption{The affine subspace $\mathbf{A}+b$ is given by the  $x$-axis $\mathbf{A}$ displaced by $b$ along the $y$-axis. The embedding  $j : \Graff(k,n) \to \Gr(k+1,n+1)$ takes $\mathbf{A} + b$ to the smallest $2$-plane containing $\mathbf{A}$ and $b+e_3$, where $e_3$ is a unit vector along the $z$-axis.}
\label{fig:plane}
\end{figure}

We remind the reader that a Grassmannian is equipped with a Radon probability measure \cite[Section~3.9]{Mattila}. All statements referring to  a measure on $\Gr(k,n)$ will be with respect to this.
\begin{theorem}\label{thm:alg}
Let $n \ge 2$ and $k \le n$.
The affine Grassmannian 
$\Graff(k,n)$ is an open submanifold of $\Gr(k+1,n+1)$ whose complement has codimension at least two and measure zero. For concreteness, we will use the map 
\begin{equation}\label{eq:j}
j : \Graff(k,n) \to \Gr(k+1,n+1), \quad \mathbf{A}+b \mapsto \spn(\mathbf{A}\cup \{ b +e_{n+1} \} ),
\end{equation}
where $e_{n+1} = (0,\dots,0,1)^\tp \in \mathbb{R}^{n+1}$ as our default embedding map.
\end{theorem}
\begin{proof}
We will prove that $j$ as defined in \eqref{eq:j} is an embedding and its image is an open subset of $\Graff(k,n)$. Let $\mathbf{A}+b \in\Graff(k,n)$.  First we observe that whenever 
\[
\spn(\mathbf{A}\cup \{ b +e_{n+1} \} ) = \spn(\mathbf{A}'\cup \{ b' +e_{n+1} \} ),
\]
we have that $b' + e_{n+1}\in \spn(\mathbf{A}\cup \{ b +e_{n+1} \} )$. This implies that $b'- b \in \mathbf{A}$ since $ \mathbf{A}$ is a subspace of $\mathbb{R}^n$. So $\mathbf{A} = \mathbf{A}'$ and therefore $\mathbf{A} + b = \mathbf{A}' + b'$, i.e., the map $j$ is injective.

The smoothness of $j$ can be seen by putting orthogonal affine coordinates on $\Graff(k,n) $ and the usual choice of coordinates on $\Gr(k+1,n+1)$ where every $\mathbf{B} \in \Gr(k+1,n+1)$ is represented by an orthonormal basis $B \in \mathbb{R}^{(n+1) \times (k+1)}$ of $\mathbf{B}$. Let $\mathbf{A} + b \in \Graff(k,n)$ have orthogonal affine coordinates $[A,b_0]\in \mathbb{R}^{n \times (k+1)}$ where $A^\tp A = I$, $A^\tp b_0 = 0$, and let $a_1,\dots, a_k\in \mathbb{R}^n$ be the column vectors of $A$. By definition, $j(\mathbf{A} + b)$ is spanned by the orthonormal basis
\[
\begin{bmatrix}
a_1\\
0
\end{bmatrix},\dots, 
\begin{bmatrix}
a_k\\
0
\end{bmatrix},
\begin{bmatrix}
b_0\\
1
\end{bmatrix}.
\]
So with our choice of coordinates on $\Graff(k,n) $ and $\Gr(k+1,n+1)$, the map $j$ takes the form
\[
j([A,b_0]) = \begin{bmatrix}
A & b_0\\
0 & 1
\end{bmatrix},
\]
which is clearly smooth.

Since $n\ge 2$, $j\bigl(\Graff(k,n)\bigr)$ is an open submanifold of $\Gr(k+1,n+1)$. For the complement of $j\bigl(  \Graff(k,n) \bigr)$ in $\Gr(k+1,n+1)$, note that a $(k+1)$-dimensional linear subspace $\mathbf{B}$ of $\mathbb{R}^{n+1}$ is an image of some $\mathbf{A} + b\in \Graff(k,n)$ under the map $j$ if and only if the $(n+1)$th coordinate of any vector in $\mathbf{B}$ is nonzero. So  the complement  consists of all $\mathbf{B}$  contained in the subspace $\mathbb{R}^n$ of $\mathbb{R}^{n+1}$, i.e., it is diffeomorphic to the Grassmannian $\Gr(k+1,n)$, which has dimension $(k+1)(n-k-1) = (k+1)(n-k) - (k+1)$, thus codimension $k+1 \ge 2$, and therefore of measure zero.
\end{proof}
In the proof we identified $\mathbb{R}^n$ with the subset $\{ (x_1,\dots,x_n,0)^\tp  \in \mathbb{R}^{n+1} : x_1,\dots,x_n \in \mathbb{R}\}$ to obtain a complete flag
$\{0\} \subset \mathbb{R}^1 \subset \mathbb{R}^2 \subset \dots \subset \mathbb{R}^n \subset \mathbb{R}^{n+1} \subset \cdots$.
Given this, our choice of $e_{n+1}$ in the embedding $j$ in \eqref{eq:j} is the most natural one. Henceforth we will often  identify $\Graff(k,n)$ with its embedded image $j\bigl(\Graff(k,n)\bigr)$. Whenever we speak of  $\Graff(k,n)$ as if it is a subset of $\Gr(k+1,n+1)$, we are implicitly assuming this identification. In this regard, we may view $\Gr(k+1,n+1)$ as a compactification of the noncompact manifold $\Graff(k,n)$.

From Theorem~\ref{thm:alg}, we derive a few other observations that will be of importance for our optimization algorithms. From the perspective of optimization, the most important feature of the embedding $j$ is that it does not increase dimension; since the computational costs of optimization algorithms inevitably depend on the dimension of the ambient space, it is ideal in this regard.
\begin{corollary}\label{prop:smooth}
$\Graff(k,n)$ is a Riemannian manifold with the canonical metric induced from that of $\Gr(k+1,n+1)$. In addition,
\begin{enumerate}[\upshape (i)]
\item\label{it:dim} the dimension of the ambient manifold $\Gr(k+1,n+1)$ is exactly the same as $\Graff(k,n)$, i.e.,
\[
\dim \Graff(k,n) = (n-k)(k+1) = \dim \Gr(k+1,n+1);
\]

\item\label{it:gd=} the geodesic distance between two points $\mathbf{A} + b$ and $\mathbf{B} + c$ in $\Graff(k,n)$ is equal to that between $j(\mathbf{A} + b)$ and $j(\mathbf{B} + c)$ in $\Gr(k+1,n+1)$;

\item\label{it:ext}  if $f : \Graff(k,n) \to \mathbb{R}$ is a continuous function that can be extended to $\widetilde{f} : \Gr(k+1,n+1)\to \mathbb{R}$, then the minimizer and maximizer of $\widetilde{f}$ are almost always attained in $\Graff(k,n)$.
\end{enumerate}
\end{corollary}
\begin{proof}
It is a basic fact in differential geometry \cite[Chapter~8]{Lee2003} that every open subset of a Riemannian manifold is also a Riemannian manifold with the induced metric. Explicit expressions for the Riemannian metric on $\Graff(k,n)$ can be found in Propositions~\ref{prop:basicstief}\eqref{it:metric} and \ref{prop:basicproj}\eqref{Pmetric}.  \eqref{it:dim} follows from Theorem~\ref{thm:alg}, i.e., $\Graff(k,n)$ is an open submanifold of $\Gr(k+1,n+1)$. Since the codimension of the complement of $\Graff(k,n)$ in $\Gr(k+1,n+1)$ is at least two, \eqref{it:gd=} follows from the Transversality Theorem in differential topology \cite{Hirsch}. For \eqref{it:ext}, note that $\widetilde{f}$ always attains its  minimizer and maximizer since $\Gr(k+1,n+1)$ is compact; that the minimizer and maximizer lie in $j\bigl(\Graff(k,n)\bigr) $ with probability one is just a consequence of the fact that its complement has null measure.
\end{proof}
Note that \eqref{it:gd=} and \eqref{it:ext}  rely  on Theorem~\ref{thm:alg} and does not hold in general for other embedded manifolds. For example, if $B$ is the solid unit ball in $\mathbb{R}^3$ and $M$ is the complement of $B$ in $\mathbb{R}^3$, then \eqref{it:gd=} and \eqref{it:ext} obviously fail to hold for $M$. 

It is sometimes desirable to represent elements of $\Gr(k,n)$ as actual matrices instead of equivalence classes of matrices. 
The Grassmannian has a well-known representation  \cite[Example~1.2.20]{Nicolaescu}  as  rank-$k$ orthogonal projection\footnote{A projection matrix satisfies $P^2 = P$ and an orthogonal projection matrix is in addition symmetric, i.e., $P^\tp  = P$. An orthogonal projection matrix $P$ is not an orthogonal matrix unless $P=I$.} matrices, or, equivalently,  trace-$k$ idempotent symmetric matrices:
\begin{equation}\label{eq:grproj}
\Gr(k,n) \cong \{P \in \mathbb{R}^{n \times n}: P^\tp  = P^{2} = P, \; \tr(P) = k\}.
\end{equation}
Note that $\rank(P) = \tr(P)$ for an orthogonal projection matrix $P$.
A straightforward affine analogue of \eqref{eq:grproj} for $\Graff(k,n)$ is simply
\begin{equation}\label{eq:graffmatrix}
\Graff(k,n) \cong \{[P , b] \in \mathbb{R}^{n \times (n+1)}:  P^\tp =P^2=P, \; \tr(P)=k,\; Pb = 0\},
\end{equation}
where $\mathbf{A}+b \in \Graff(k,n)$ with orthogonal affine coordinates $[A,b_0] \in \mathbb{R}^{n \times (k+1)}$ is represented as the matrix\footnote{If $A$ is an orthonormal basis for the subspace $\mathbf{A}$, then $AA^\tp $ is the orthogonal projection onto $\mathbf{A}$.} $[AA^\tp , b_0] \in \mathbb{R}^{n \times (n+1)}$. We  call this  the matrix of \emph{projection affine coordinates} for $\mathbf{A} + b$.

There are three particularly useful systems of matrix coordinates on the Grassmannian: a point on $\Gr(k,n)$ can be represented as (i) an equivalence class of matrices  $A \in \mathbb{R}^{n \times k}$ with linearly independent columns such that $A \sim AS$ for any $S \in \GL(k)$, the group of invertible $k\times k$ matrices; (ii) an equivalence class of matrices $A \in \V(k,n)$ with orthonormal columns such that $A \sim AQ$ for any $Q \in \O(k)$; (iii) a projection matrix $P \in \mathbb{R}^{n \times n}$ satisfying $P^2 =P^\tp  = P$ and $\tr(P) = k$. These correspond to representing  $\mathbf{A}$ by (i) bases of $\mathbf{A}$, (ii) orthonormal bases  of $\mathbf{A}$, (iii) an orthogonal projection onto $\mathbf{A}$. The affine coordinates, orthogonal affine coordinates, and projection affine coordinates introduced in this section are obvious analogues of (i), (ii), and (iii) respectively.

However, these relatively simplistic global coordinates are inadequate in computations. As we will see in Sections~\ref{sec:objects} and \ref{sec:optim}, explicit representations of tangent space vectors and geodesics, effective computations of exponential maps, parallel transports, gradients, and Hessians, require more sophisticated systems of global matrix coordinates. In Section~\ref{sec:global} we will introduce two of these.

Nevertheless, the simpler coordinate systems in this section serve a valuable role ---  they come in handy in proofs, where the more complicated systems of coordinates in Section~\ref{sec:global} can be unnecessarily cumbersome. The bottom line is that different coordinates are good for different purposes\footnote{This is also the case for Grassmannian: orthonormal or projection matrix coordinates may be invaluable for computations as in \cite{EAS, newgrass} but they obscure mathematical properties evident in, say, Pl\"ucker coordinates \cite[Chapter~14]{MS}.} and having several choices makes the affine Grassmannian a versatile platform in applications.

\section{Matrix coordinates for the affine Grassmannian}\label{sec:global}

One reason for the wide applicability of the Grassmannian is the existence of  several excellent choices of global coordinates in terms of matrices, allowing subspaces to be represented as matrices and thereby facilitating the use of a vast range of algorithms in numerical linear algebra \cite{AMSbook, AMS, AMSV, EAS}. Here we will introduce two systems of global coordinates, representing a point on $\Graff(n,k)$ as an $(n+1) \times (k+1)$ orthonormal matrix or as an $(n+1) \times (n+1)$ projection matrix respectively.

For  an affine subspace $\mathbf{A} + b \in\Graff(k,n)$, its orthogonal affine coordinates are $[A,b_0] \in \V(k,n) \times \mathbb{R}^n$ where $A^\tp  b_0 = 0$, i.e., $b_0$ is orthogonal to the columns of $A$. However as $b_0$ is in general not of unit norm, we may not regard $[A,b_0]$ as an element of $\V(n,k+1)$. With this in mind, we introduce the notion of Stiefel coordinates, which is the most suitable system of coordinates for computations.
\begin{definition}\label{def:stief}
Let $\mathbf{A} + b \in \Graff(k,n)$  and $[A,b_0] \in \mathbb{R}^{n \times (k+1)}$ be its orthogonal affine coordinates, i.e., $A^\tp  A = I$ and $A^\tp  b_0 =0$. The matrix of \emph{Stiefel coordinates} for $\mathbf{A} + b$ is the $(n+1) \times (k+1)$ matrix with orthonormal columns,
\[
Y_{\mathbf{A} + b} \coloneqq
\begin{bmatrix}
A& b_0/\sqrt{1+\lVert b_0\rVert^2}\\
0& 1/\sqrt{1+\lVert b_0\rVert^2}
\end{bmatrix} \in \V(n +1, k+1).
\]
\end{definition}
Two orthogonal affine coordinates $[A,b_0], [A',b_0']$ of $\mathbf{A}+b$ give two corresponding matrices of Stiefel coordinates $Y_{\mathbf{A} + b} $, $Y_{\mathbf{A} + b}'$. By the remark after our definition of orthogonal affine coordinates, $A = A'Q'$ for some $Q' \in \O(k)$ and $b_0 = b_0'$. Hence
\begin{equation}\label{eq:otsc}
Y_{\mathbf{A} + b} =
\begin{bmatrix}
A& b_0/\sqrt{1+\lVert b_0\rVert^2}\\
0& 1/\sqrt{1+\lVert b_0\rVert^2}
\end{bmatrix}= 
\begin{bmatrix}
A' & b_0'/\sqrt{1+\lVert b_0'\rVert^2}\\
0& 1/\sqrt{1+\lVert b_0'\rVert^2}
\end{bmatrix}
\begin{bmatrix}
Q' & 0 \\
0 & 1
\end{bmatrix} =
Y_{\mathbf{A} + b}' Q
\end{equation}
where $Q \coloneqq \begin{bsmallmatrix} Q' & 0 \\ 0 &1\end{bsmallmatrix}\in \O(k+1)$. Hence two different matrices of Stiefel coordinates for the same affine space differ by an orthogonal transformation.

\begin{proposition} \label{prop:graffquot}
Consider the equivalence class of matrices given by
\[
\begin{bmatrix} A& b \\ 0&  \gamma \end{bmatrix} \cdot \O(k+1) \coloneqq \biggl\{ \begin{bmatrix} A& b \\ 0&  \gamma \end{bmatrix}Q  \in \mathbb{R}^{(n+1) \times (k+1)} :  Q \in \O(k+1) \biggr\}.
\]
The affine Grassmannian may be represented as a set of equivalence classes of $(n+1) \times (k+1)$ matrices with orthonormal columns,
\begin{align}
\Graff(k,n) &\cong \biggl\{ 
\begin{bmatrix} A& b \\ 0&  \gamma \end{bmatrix} \cdot \O(k+1)  :  \begin{bmatrix} A& b \\ 0&  \gamma \end{bmatrix} \in \V(k+1, n+1)\biggr\} \label{eq:X1} \\
&\subseteq \V(k+1, n+1)/\O(k+1) = \Gr(k+1,n+1). \label{eq:X2}
\end{align}
An affine subspace $\mathbf{A} + b \in \Graff(k,n)$  is represented by the equivalence class $Y_{\mathbf{A} + b} \cdot \O(k+1)$ corresponding to its matrix of Stiefel coordinates.
\end{proposition} 
\begin{proof}
The set of equivalence classes on the \textsc{rhs} of \eqref{eq:X1} is the set $X$ in Theorem~\ref{thm:alg}(iii) if $\Gr(k+1,n+1)$ is regarded as the homogeneous space in \eqref{eq:X2}.
\end{proof}

The following lemma is easy to see from the definition of Stiefel coordinates and our discussion above. It will be useful for the optimization algorithms in Section~\ref{sec:optim}, allowing us to check \emph{feasibility}, i.e., whether a point represented as an $(n+1) \times (k+1)$ matrix is in the feasible set $j\bigl(\Graff(k,n)\bigr)$.
\begin{lemma}\label{lem:checkgraffstief}
\begin{enumerate}[\upshape (i)]
\item Any matrix of the form $\begin{bsmallmatrix} A& b \\ 0&  \gamma \end{bsmallmatrix} \in \V(k+1, n+1)$, i.e.,
\[
A^\tp A = I, \qquad A^\tp  b = 0, \qquad \lVert b\rVert^2 + \gamma^2 = 1,
\]
is the matrix of Stiefel coordinates for some $\mathbf{A} + b \in \Graff(k,n)$.
\item\label{two} Two matrices of Stiefel coordinates $\begin{bsmallmatrix} A& b \\ 0&  \gamma \end{bsmallmatrix}, \begin{bsmallmatrix} A'& b' \\ 0&  \gamma' \end{bsmallmatrix} \in \V(k+1, n+1)$  represent the same affine subspace iff there exists $\begin{bsmallmatrix} Q' & 0 \\ 0 &1\end{bsmallmatrix} \in \O(k+1)$ such that
\[
\begin{bmatrix} A& b \\ 0&  \gamma \end{bmatrix} = 
\begin{bmatrix} A'& b' \\ 0&  \gamma' \end{bmatrix}  
\begin{bmatrix} Q'& 0 \\ 0&  1\end{bmatrix}.
\]
\item If $\begin{bsmallmatrix} A& b \\ 0&  \gamma \end{bsmallmatrix} \in \V(k+1, n+1)$ is a matrix of Stiefel coordinates for $\mathbf{A} + b$, then every other matrix of Stiefel coordinates for $\mathbf{A} + b$ belongs to the equivalence class  $\begin{bsmallmatrix} A& b \\ 0&  \gamma \end{bsmallmatrix} \cdot \O(k+1)$, but not every matrix in  $\begin{bsmallmatrix} A& b \\ 0&  \gamma \end{bsmallmatrix} \cdot \O(k+1)$ is a matrix of Stiefel coordinates for $\mathbf{A} + b$.
\end{enumerate}
\end{lemma}

The matrix of projection affine coordinates $[P , b] \in \mathbb{R}^{n \times (n+1)}$ in \eqref{eq:graffmatrix} is not an orthogonal projection matrix. With this in mind, we introduce the following notion.
\begin{definition}\label{def:proj}
Let $\mathbf{A} + b \in \Graff(k,n)$  and $[P , b] \in \mathbb{R}^{n \times (n+1)}$  be its projection affine coordinates. The matrix of \emph{projection coordinates} for $\mathbf{A} + b$ is the  orthogonal projection matrix
\[
P_{\mathbf{A} + b} \coloneqq \begin{bmatrix} P+bb^\tp /(\lVert b \rVert^{2}+1)& b/(\lVert b\rVert^{2}+1) \\ b^\tp /(\lVert b\rvert^{2}+1)& 1/(\lVert b\rVert^{2}+1) \end{bmatrix} \in \mathbb{R}^{(n+1) \times (n+1)}.
\]
Alternatively, in terms of orthogonal affine coordinates $[A,b_0] \in \mathbb{R}^{n \times (k+1)}$,
\[
P_{\mathbf{A} + b} = \begin{bmatrix} AA^\tp  +b_0b_0^\tp /(\lVert b_0 \rVert^{2}+1)& b_0/(\lVert b_0\rVert^{2}+1) \\ b_0^\tp /(\lVert b_0\rvert^{2}+1)& 1/(\lVert b_0\rVert^{2}+1)\end{bmatrix} \in \mathbb{R}^{(n+1) \times (n+1)}.
\]
\end{definition}
It is straightforward to verify that $P_{\mathbf{A} + b}$ is indeed an orthogonal projection matrix, i.e.,
$P_{\mathbf{A} + b}^2 =  P_{\mathbf{A} + b} = P_{\mathbf{A} + b}^\tp $.
Unlike Stiefel coordinates, projection coordinates of a given affine subspace are unique.
As in Proposition~\ref{prop:graffquot}, the next result gives a concrete description of  the set $X = j \bigl(\Graff(k,n)\bigr)$ in Theorem~\ref{thm:alg}(iii), but  in terms of projection coordinates. With this description, $\Graff(k,n)$ may be regarded as a subvariety of $\mathbb{R}^{(n+1) \times (n+1)}$.
\begin{proposition}\label{prop:graffproj}
The affine Grassmannian may be represented as  a set of $(n+1) \times (n+1)$ orthogonal projection matrices,
\begin{multline}\label{eq:graffproj}
\Graff(k,n) \cong \biggl\{\begin{bmatrix} P+bb^\tp /(\lVert b \rVert^{2}+1)& b/(\lVert b\rVert^{2}+1) \\ b^\tp /(\lVert b\rvert^{2}+1)& 1/(\lVert b\rVert^{2}+1) \end{bmatrix} \in \mathbb{R}^{(n+1) \times (n+1)} :  \\
P \in \mathbb{R}^{n \times n}, \; P^\tp  = P^{2} = P, \; \tr(P) = k, \; Pb = 0 \biggr\}.
\end{multline}
An affine subspace $\mathbf{A} + b \in \Graff(k,n)$ is uniquely represented by its projection coordinates $P_{\mathbf{A} + b}$.
\end{proposition}
\begin{proof}
Let $\mathbf{A} + b \in \Graff(k,n)$ have orthogonal affine coordinates $[A,b_0]$. Since $P = AA^\tp  \in \mathbb{R}^{n \times n}$ is an orthogonal projection matrix that satisfies $Pb_0 = 0$, the map $\mathbf{A} + b \mapsto P_{\mathbf{A} + b}$  takes $\Graff(k,n)$ onto the set of matrices on the \textsc{rhs} of \eqref{eq:graffproj} with inverse given by $ P_{\mathbf{A} + b} \mapsto \im(P) + b_0$.
\end{proof}

The next lemma allows feasibility checking in projection coordinates.
\begin{lemma}\label{lem:checkgraffproj}
An orthogonal projection matrix $\begin{bsmallmatrix} S & d \\ d^\tp  & \gamma \end{bsmallmatrix} \in \mathbb{R}^{(n+1) \times (n+1)}$ is the matrix of projection coordinates for some affine subspace in $\mathbb{R}^n$ iff
\begin{enumerate}[\upshape (i)]
\item $\gamma \ne 0$;
\item $S - \gamma^{-1} dd^\tp  \in \mathbb{R}^{n \times n}$ is an orthogonal projection matrix;
\item $S d = 0$.
\end{enumerate}
In addition, $\begin{bsmallmatrix} S & d \\ d^\tp  & \gamma \end{bsmallmatrix}  \in \mathbb{R}^{(n+1) \times (n+1)}$ is the matrix of projection coordinates for $\mathbf{A} + b \in \Graff(k,n)$ iff $S - \gamma^{-1} dd^\tp  = AA^\tp $ and $\gamma^{-1} d = b_0$ where $[A,b_0] \in \mathbb{R}^{n \times (k+1)}$ is $\mathbf{A} + b$ in orthogonal affine coordinates.
\end{lemma}

The next lemma allows us to switch between Stiefel and projection  coordinates.
\begin{lemma}\label{lem:interchange}
\begin{enumerate}[\upshape (i)]
\item\label{StoP} If $Y_{\mathbf{A} + b} \in \V(k+1,n+1)$ is a matrix of Stiefel coordinates for $\mathbf{A} + b$, then the matrix of projection coordinates for $\mathbf{A} + b$ is given by
\[
P_{\mathbf{A} + b} =  Y_{\mathbf{A} + b} Y_{\mathbf{A} + b}^\tp  \in \mathbb{R}^{(n+1) \times (n+1)}.
\]

\item\label{PtoS} If $P_{\mathbf{A} + b}  \in \mathbb{R}^{(n+1) \times (n+1)}$ is the matrix of projection coordinates for $\mathbf{A} + b$, then a matrix of Stiefel coordinates for $\mathbf{A} + b$ is given by any $Y_{\mathbf{A} + b} \in \V(k+1,n+1)$ whose columns form an orthonormal eigenbasis for the $1$-eigenspace of $P_{\mathbf{A} + b}$.
\end{enumerate}
\end{lemma}
\begin{proof}
\eqref{StoP} follows from the observation that for any  $Q \in \O(k+1)$,
{\footnotesize
\[
\biggl(\begin{bmatrix} A& b/\sqrt{\lVert b \rVert^{2}+1} \\ 0& 1/\sqrt{ \lVert b\rVert^{2}+1} \end{bmatrix} Q\biggr)\biggl(\begin{bmatrix} A& b/\sqrt{\lVert b \rVert^{2}+1} \\ 0& 1/\sqrt{ \lVert b\rVert^{2}+1} \end{bmatrix} Q\biggr)^\tp 
= \begin{bmatrix}AA^\tp + bb^\tp /\lVert b \rVert^{2}+1 & b/(\lVert b\rVert^{2}+1) \\ b^\tp /(\lVert b\rvert^{2}+1)& 1/(\lVert b\rVert^{2}+1) \end{bmatrix}.
\]}%
For \eqref{PtoS}, recall that the eigenvalues of an orthogonal projection matrix are $0$'s and $1$'s with multiplicities given by its nullity and rank respectively. Thus we have an eigenvalue decomposition of the form
$P_{\mathbf{A} + b}   = V \begin{bsmallmatrix}I_{k+1} & \\ & 0_{n -k} \end{bsmallmatrix} V^\tp  = V_{k+1} V_{k+1}^\tp $,
where the columns of $V_{k+1} \in \V(k+1,n+1)$ are the eigenvectors corresponding to the eigenvalue $1$. Let $v\in \mathbb{R}^{k+1}$ be the last row of $V_{k+1}$ and $Q \in \O(k+1)$ be a Householder matrix \cite{GVL} such that  $Q^\tp  v = \lVert v\rVert e_{k+1}$. Then $Y_{\mathbf{A} + b} = V_{k+1}Q$ has the form required in Lemma~\ref{lem:checkgraffstief}(i) for a matrix of Stiefel coordinates. 
\end{proof}

The above proof also shows that projection coordinates are unique even though Stiefel coordinates are not. In principle, they are interchangeable via Lemma~\ref{lem:interchange} but in reality, one form is usually more natural than the other for a specific use.

\section{Tangent space, exponential map, geodesic, parallel transport, gradient, and Hessian on the affine Grassmannian}\label{sec:objects}

The embedding of $\Graff(k,n)$ as an open smooth submanifold of $\Gr(k+1,n+1)$ by Theorem~\ref{thm:alg} and Corollary~\ref{prop:smooth} allows us to borrow the  Riemannian optimization framework on Grassmannians in \cite{AMSbook, AMS, AMSV, EAS} to develop optimization algorithms on the affine Grassmannian. 
We will present various geometric notions and algorithms on $\Graff(k,n)$ in terms of both Stiefel and projection coordinates. The higher dimensions required by projection coordinates generally makes them  less desirable than  Stiefel coordinates.

Propositions~\ref{prop:basicstief}, Theorem~\ref{thm3}, and Proposition~\ref{prop:basicproj} are respectively summaries of \cite{EAS} and \cite{newgrass}  adapted for the affine Grassmannian. We will only give a sketch of the proof, referring readers to the original sources for more details. 
\begin{proposition}\label{prop:basicstief}
The following are basic differential geometric notions on $\Graff(k,n)$ expressed in Stiefel coordinates.
\begin{enumerate}[\upshape (i)]
\item Tangent space: The tangent space at $\mathbf{A}+b \in \Graff(k,n)$ has representation
\[
\T_{\mathbf{A}+b}\bigl(\Graff(k,n)\bigr) = \bigl\{\Delta \in \mathbb{R}^{(n+1) \times (k+1)}: Y_{\mathbf{A}+b}^\tp \Delta = 0 \bigr\}.
\] 

\item\label{it:metric} Riemannian metric: The Riemannian metric $g$ on $\Graff(k,n)$ is given by
\[
g_{\mathbf{A} + b} (\Delta_1,\Delta_2) = \tr(\Delta_1^\tp \Delta_2)
\]
for $\Delta_1, \Delta_2 \in \T_{\mathbf{A}+b}\bigl(\Graff(k,n)\bigr) $, i.e., $Y_{\mathbf{A}+b}^\tp \Delta_i = 0$, $i =1,2$.

\item\label{Exp} Exponential map:
The geodesic with $Y(0) = Y_{\mathbf{A}+b}$ and $\dot{Y}(0) = H$ in $\Graff(k,n)$ is given by
\[
Y(t) = \begin{bmatrix} Y_{\mathbf{A}+b}V&U \end{bmatrix} \begin{bmatrix} \cos (t\Sigma)\\ \sin (t\Sigma) \end{bmatrix}V^\tp ,
\]
where $H = U\Sigma V^\tp $ is a condensed \textsc{svd}.

\item\label{Par} Parallel transport:
The parallel transport of $\Delta \in \T_{\mathbf{A}+b}\bigl(\Graff(k,n)\bigr)$ along the geodesic given by $H$ has expression
\[
\tau \Delta(t) = \biggl(\begin{bmatrix} Y_{\mathbf{A}+b}V&U \end{bmatrix} \begin{bmatrix} -\sin (t\Sigma)\\ \cos (t\Sigma) \end{bmatrix} U^\tp  + (I-UU^\tp ) \biggr) \Delta,
\]
where $H = U\Sigma V^\tp $ is a  condensed \textsc{svd}.

\item\label{Sgrad} Gradient: 
Let $f: \mathbb{R}^{(n+1) \times (k+1)} \rightarrow \mathbb{R}$ satisfy $f(YQ) = f(Y)$  for every $Y$ with  $Y^\tp Y = I$ and $Q \in \O(k+1)$. The gradient of $f$  at  $Y = Y_{\mathbf{A}+b}$ is
\[
\nabla f = f_{Y} - YY^\tp f_Y \in \T_{\mathbf{A}+b}\bigl(\Graff(k,n)\bigr),
\]
where $f_Y \in \mathbb{R}^{(n+1) \times (k+1)}$ with $(f_Y)_{ij} = \frac{\partial f}{\partial y_{ij}}$.

\item Hessian:
Let $f$ be as in \eqref{Sgrad}. The Hessian of $f$  at  $Y = Y_{\mathbf{A}+b}$ is
\begin{enumerate}[\upshape (a)]
\item as a bilinear form: $\nabla^2 f : \T_{\mathbf{A}+b}\bigl(\Graff(k,n)\bigr) \times \T_{\mathbf{A}+b}\bigl(\Graff(k,n)\bigr) \rightarrow \mathbb{R}$,
\[
\nabla^2 f(\Delta, \Delta') = f_{YY}(\Delta, \Delta') - \tr(\Delta^\tp \Delta'Y^\tp f_{Y}),
\]
where $f_{YY} \in \mathbb{R}^{(n+1)  (k+1) \times (n+1)  (k+1)}$ with  $(f_{YY})_{ij,hl} = \frac{\partial^{2} f}{\partial y_{ij} \partial y_{hl}}$ and
\[
 f_{YY}(\Delta,\Delta') = \sum\nolimits_{i,j,h,l=1}^{n+1,k+1,n+1,k+1} (f_{YY})_{ij,hl} \delta_{ij} \delta'_{hl};
\]

\item as a linear map: $ \nabla^2  f:  \T_{\mathbf{A}+b}\bigl(\Graff(k,n)\bigr) \rightarrow \T_{\mathbf{A}+b}\bigl(\Graff(k,n)\bigr)$,
\[
\nabla^2 f(\Delta) = \sum\nolimits_{i,j,h,l=1}^{n+1,k+1,n+1,k+1} (f_{YY})_{ij,hl}\delta_{ij} E_{hl} - \Delta f_Y^\tp  Y,
\]
where $E_{hl} \in \mathbb{R}^{(n+1) \times (k+1)}$ has $(h,l)$th entry $1$ and all other entries $0$.
\end{enumerate}
\end{enumerate}
\end{proposition}
\begin{proof}[Sketch of proof]
These essentially follow from the corresponding formulas for the Grassmannian in \cite{EAS, newgrass}. For instance, the Riemannian metric $g$ is induced by the canonical Riemannian metric  on $\Gr(k+1,n+1)$ \cite[Section~2.5]{EAS}, the geodesic $X(t)$ on $\Gr(k,n)$ starting at $X(0) =X_{\mathbf{A}}$ in the direction $\dot{X}(0) = H$ is given in \cite[Equation (2.65)]{EAS} as
\[
X(t) = \begin{bmatrix}
X_{\mathbf{A}} & U
\end{bmatrix} 
\begin{bmatrix}
\cos (t\Sigma) \\ \sin(t\Sigma)
\end{bmatrix} V^\tp,
\]
where $X_{\mathbf{A}}$ is the matrix representation of $X(0)$ and $H = U \Sigma V^\tp$ is a  condensed \textsc{svd} of $H$. The displayed formula in \eqref{Exp} for a geodesic in $\Graff(k,n)$ is then obtained by taking the inverse image of the corresponding geodesic in $\Gr(k+1,n+1)$ under the embedding $j$. Other formulas may be similarly obtained by the same procedure from their counterparts on the Grassmannian.
\end{proof}

Since the distance minimizing geodesic connecting two points on $\Gr(k+1,n+1)$ is not necessarily unique,\footnote{For example, there are two distance minimizing geodesics on $\Gr(1,2) \simeq \mathbb{S}^1$ for any pair of antipodal points.}
it is possible that there is more than one geodesic on $\Graff(k,n)$ connecting two given points. However, distance minimizing geodesics can all be parametrized as in Proposition \ref{prop:basicstief} even if they are not unique. In fact, we may explicitly compute the geodesic distance between any two points on $\Graff(k,n)$ as follows.
\begin{theorem}\label{thm3}
For any two affine $k$-flats $\mathbf{A}+b$ and $\mathbf{B}+c \in \Graff(k,n)$,
\[
d_{\Graff(k,n)}(\mathbf{A}+b,\mathbf{B}+c) \coloneqq d_{\Gr(k+1,n+1)}\bigl(j(\mathbf{A}+b), j(\mathbf{B}+c) \bigr),
\]
where $j$ is the embedding in \eqref{eq:j}, defines a notion of distance consistent with the geodesic distance on a Grassmannian.  If
\[
Y_{\mathbf{A} + b} =
\begin{bmatrix}
A& b_0/\sqrt{1+\lVert b_0\rVert^2}\\
0& 1/\sqrt{1+\lVert b_0\rVert^2}
\end{bmatrix},\qquad
Y_{\mathbf{B} + c}=
\begin{bmatrix}
B& c_0/\sqrt{1+\lVert c_0\rVert^2}\\
0& 1/\sqrt{1+\lVert c_0\rVert^2}
\end{bmatrix}
\]
are the matrices of Stiefel coordinates for $\mathbf{A}+b$ and $\mathbf{B}+c$ respectively, then
\begin{equation}\label{eq:graffdist}
d_{\Graff(k,n)}(\mathbf{A}+b,\mathbf{B}+c) = \Bigl(\sum\nolimits_{i=1}^{k+1} \theta_i^2\Bigr)^{1/2},
\end{equation}
where $\theta_i = \cos^{-1} \sigma_i$ and $\sigma_i$ is the $i$th singular value of $Y_{\mathbf{A} + b}^\tp  Y_{\mathbf{B} + c} \in \mathbb{R}^{(k+1) \times (k+1)}$.
\end{theorem}
\begin{proof}
Any nonempty subset of a metric space is a metric space. It remains to check that the definition does not depend on a choice of Stiefel coordinates. Let $Y_{\mathbf{A} + b}$ and $Y_{\mathbf{A} + b}' $ be two different matrices of Stiefel coordinates for $\mathbf{A} + b$ and  $Y_{\mathbf{B} + c}$ and $Y_{\mathbf{B} + c}' $ be two different matrices of Stiefel coordinates for $\mathbf{B} + c$. By Lemma~\ref{lem:checkgraffstief}\eqref{two}, there exist $Q_1, Q_2 \in \O(k+1)$ such that
$Y_{\mathbf{A} + b} =Y_{\mathbf{A} + b}' Q_1$, $Y_{\mathbf{B} + c} =Y_{\mathbf{B} + c}' Q_2$.
The required result then follows from
\[
\sigma_i(Y_{\mathbf{A} + b}^\tp  Y_{\mathbf{B} + c}) = \sigma_i(Q_1^\tp  Y_{\mathbf{A} + b}^{\prime \tp} Y_{\mathbf{B} + c}' Q_2) = \sigma_i(Y_{\mathbf{A} + b}^{\prime \tp} Y_{\mathbf{B} + c}'), \qquad i=1,\dots,k. \qedhere
\]
\end{proof}
The proof above also shows that $\theta_1,\dots,\theta_{k+1}$ are independent of the choice of Stiefel coordinates. We will call $\theta_i $ the $i$th \emph{affine principal angles} between the respective affine subspaces and denote it by $\theta_i(\mathbf{A}+b,\mathbf{B}+c)$. Consider the \textsc{svd},
\begin{equation}\label{eq:scsvd}
 Y_{\mathbf{A} + b}^\tp  Y_{\mathbf{B} + c} = U \Sigma V^\tp 
\end{equation}
where $U,V \in \O(k+1)$ and $\Sigma = \diag (\sigma_1,\dots,\sigma_{k+1})$. Let
\[
 Y_{\mathbf{A} + b} U = [p_1,\dots,p_{k+1}], \qquad Y_{\mathbf{B} + c} V = [q_1,\dots,q_{k+1}].
\]
We will call the pair of column vectors $(p_i,q_i)$ the $i$th \emph{affine principal vectors} between $\mathbf{A}+b$ and $\mathbf{B}+c$. These are the affine analogues of principal angles and vectors of linear subspaces \cite{BG, GVL, YL}. 

This expression for a geodesic in Proposition~\ref{prop:basicstief}\eqref{Exp} assumes that we are given an initial point and an initial direction, the following gives an alternative expression for a geodesic in $\Graff(k,n)$ that connects two given points.
\begin{corollary} \label{corollary:geodesic}
Let $\mathbf{A}+b$ and $\mathbf{B}+c\in \Graff(k,n)$. Let $\gamma : [0,1]\to \Gr(k+1,n+1)$ be the curve
\begin{equation}\label{eq:geodesic}
\gamma(t) = \spn(Y_{\mathbf{A} + b} U \cos (t\Theta) U^\tp + Q\sin (t\Theta) U^\tp),
\end{equation}
where $Q,U \in \O(k+1)$ and the diagonal matrix $\Theta \in \mathbb{R}^{(k+1) \times (k+1)}$ are determined by the \textsc{svd}
\[
(I - Y_{\mathbf{A} + b} Y_{\mathbf{A} + b}^\tp ) Y_{\mathbf{B} + c} (Y_{\mathbf{A} + b}^\tp  Y_{\mathbf{B} + c})^{-1} = Q (\tan \Theta) U^\tp .
\]
The orthogonal matrix $U$ is the same as that in \eqref{eq:scsvd} and $\Theta = \diag(\theta_1,\dots, \theta_{k+1})$ is the diagonal matrix of affine principal angles. Then $\gamma$ has the following properties:
\begin{enumerate}[\upshape (i)]
\item $\gamma$ is a distance minimizing curve connecting $j(\mathbf{A}+b)$ and $j(\mathbf{B}+c)$, i.e.,  attains  \eqref{eq:graffdist};
\item the derivative of $\gamma$ at $t=0$ is given by
\begin{equation}\label{eq:log}
\gamma'(0) =  Q \Theta U^\tp;
\end{equation}
\item\label{item:out} there is at most one value of $t\in (0,1)$ such that $\gamma(t) \notin j\bigl(\Graff(k,n)\bigr)$.
\end{enumerate} 
\end{corollary}
\begin{proof}[Sketch of proof]
The expression in \cite[Theorem~2.3]{EAS} for a distance minimizing geodesic connecting two points in $\Gr(k+1,n+1)$ gives \eqref{eq:geodesic}. By Theorem~\ref{thm:alg}, $\Graff(k,n)$ is embedded in $\Gr(k+1,n+1)$ as an open submanifold whose complement has measure zero. Since the complement of $\Graff(k,n)$ in $\Gr(k+1,n+1)$ comprises points with coordinates $\begin{bsmallmatrix}
A \\
0 
\end{bsmallmatrix}\in \mathbb{R}^{(n+1) \times (k+1)}$ where $A\in \mathbb{R}^{n\times (k+1)}$ and $A^\tp A = I$, a simple calculation shows that $\gamma$ has at most one point  not contained in $\Graff(k,n)$.
\end{proof}
Indeed, as the complement of $\Graff(k,n)$ in $\Gr(k+1,n+1)$ has codimension at least two, the situation $\gamma(t)\not\in \Graff(k,n)$ occurs with probability zero. For an analogue, one may think of geodesics connecting two points in $\mathbb{R}^3$ with the $x$-axis removed. This together with the proof of Theorem \ref{thm3} guarantees that  Algorithms~\ref{alg:sds}--\ref{alg:nmp} will almost never lead to a point outside $\Graff(k,n)$.

We conclude this section with the analogue of Proposition~\ref{prop:basicstief} in projection coordinates. 
\begin{proposition}\label{prop:basicproj}
The following are basic differential geometric notions on $\Graff(k,n)$ expressed in projection coordinates. We write $[X,Y] = XY-YX$ for the commutator bracket and $\mathsf{\Lambda}^2(\mathbb{R}^n)$ for the space of $n\times n$ skew symmetric matrices.
\begin{enumerate}[\upshape (i)]
\item Tangent space: The tangent space at $\mathbf{A}+b \in \Graff(k,n)$ has representation
\[
\T_{\mathbf{A} + b}\bigl(\Graff(k,n)\bigr) = \{[P_{\mathbf{A} + b},\Omega] \in \mathbb{R}^{(n+1)\times (n+1)}: \Omega \in \mathsf{\Lambda}^2(\mathbb{R}^{n+1})\}.
\]
\item\label{Pmetric} Riemannian metric: The Riemannian metric $g$ on $\Graff(k,n)$ is given by 
\[
g_{\mathbf{A} + b}(\Delta_1,\Delta_2) = \tr(\Delta_1^\tp \Delta_2),
\] 
where $\Delta_1,\Delta_2\in \T_{\mathbf{A} + b}\bigl(\Graff(k,n)\bigr) $, i.e., $\Delta_i = [P_{\mathbf{A} + b}, \Omega_i]$ for some $\Omega_i \in \mathsf{\Lambda}^2(\mathbb{R}^{n+1}), i=1,2$.

\item Exponential map:
Let $P = P_{\mathbf{A} + b}$ and $\Theta \in \mathbb{R}^{(n+1) \times (n+1)}$ be such that $[[P,\Omega],P] = \Theta^\tp  \begin{bsmallmatrix} 0&Z \\ -Z^\tp &0 \end{bsmallmatrix} \Theta$ and $P = \Theta^\tp   \begin{bsmallmatrix} I_{k+1}&0 \\ 0&0 \end{bsmallmatrix}  \Theta$. The exponential map is given by
\[
\exp_{\mathbf{A} + b}([P,\Omega]) = \frac{1}{2} I_{n+1} + \Theta^{T} \begin{bmatrix} \frac{1}{2}\cos(2\sqrt{ZZ^\tp })&-\sinc(2\sqrt{ZZ^\tp })Z \\ -Z^\tp \sinc(2\sqrt{ZZ^\tp })&-\frac{1}{2}\sin(2\sqrt{Z^\tp Z}) \end{bmatrix}  \Theta.
\]

\item\label{Pgrad} Gradient: 
Let $f: \mathbb{R}^{(n+1) \times (n+1)} \rightarrow \mathbb{R}$. The gradient of $f$  at  $P = P_{\mathbf{A}+b}$ is
\[
\nabla f = [P,[P,f_P]] \in \T_{\mathbf{A} + b}\bigl(\Graff(k,n)\bigr),
\]
where $f_P \in \mathbb{R}^{(n+1) \times (n+1)}$ with $(f_P)_{ij} = \frac{\partial f}{\partial p_{ij}}$.

\item Hessian:
Let $f$ and $f_P$ be as in \eqref{Pgrad}. The Hessian of $f$  at  $P = P_{\mathbf{A}+b}$ is
\begin{enumerate}[\upshape (a)]
\item as a bilinear form: $\nabla^2 f : \T_{\mathbf{A} + b}\bigl(\Graff(k,n)\bigr) \times \T_{\mathbf{A} + b}\bigl(\Graff(k,n)\bigr) \rightarrow \mathbb{R}$,
\[
\nabla^2 f(\Delta,\Delta') = \tr\Bigl( \bigl( \bigl[P, \bigl[P,\sum\nolimits_{i,j,h,l=1}^{n+1} (f_{PP})_{ij,hl} \delta_{ij} E_{hl}\bigr]\bigr]
- \frac{1}{2}[P,[\nabla f, \Delta]] - \frac{1}{2}[\nabla f, [P, \Delta]] \bigr)\Delta'\Bigr),
\]
where $f_{PP} \in \mathbb{R}^{(n+1)^2 \times (n+1)^2}$ with $(f_{PP})_{ij,hl} = \frac{\partial^{2} f}{\partial p_{ij} \partial p_{hl}}$ and $E_{hl} \in \mathbb{R}^{(n+1) \times (n+1)}$ has $(h,l)$th entry $1$ and all other entries $0$;

\item as a linear map: $ \nabla^2  f:  \T_{\mathbf{A} + b}\bigl(\Graff(k,n)\bigr) \rightarrow \T_{\mathbf{A} + b}\bigl(\Graff(k,n)\bigr)$,
\[
\nabla^2 f(\Delta) = \bigl[P,\bigl[P,\sum\nolimits_{i,j,h,l=1}^{n+1} (f_{PP})_{ij,hl} \delta_{ij} E_{hl}\bigr]\bigr] - \frac{1}{2}[P,[\nabla f, \Delta]] - \frac{1}{2} [\nabla f, [P, \Delta]].
\]
\end{enumerate}
\end{enumerate}
\end{proposition}
\begin{proof}[Sketch of proof]
Again, these formulas follow from their counterparts on Grassmannian manifolds in \cite{newgrass} by applying $j^{-1}$, as we did in the proof of Proposition~\ref{prop:basicstief}.
\end{proof}
A notable omission in Proposition~\ref{prop:basicproj} is a formula for parallel transport. While parallel transport on $\Graff(k,n)$ in Stiefel coordinates takes a relatively simple form in Proposition~\ref{prop:basicstief}, its explicit expression in projection coordinates is extremely complicated, and as a result unilluminating and error-prone. We do not recommend computing parallel transport in projection coordinates --- one should instead change projection coordinates to Stiefel coordinates by Lemma~\ref{lem:interchange}, compute parallel transport in Stiefel coordinates using Proposition~\ref{prop:basicstief}\eqref{Par}, and then transform the result back to projection coordinates by Lemma~\ref{lem:interchange} again. 

\section{Steepest descent, conjugate gradient, and Newton method on the affine Grassmannian}\label{sec:optim}

We now describe the methods of steepest descent, conjugate gradient, and Newton on the affine Grassmannian. The steepest descent and Newton methods are given in both Stiefel coordinates (Algorithms~\ref{alg:sds} and \ref{alg:nms}) and projection coordinates (Algorithms~\ref{alg:sdp} and \ref{alg:nmp}) but the conjugate gradient method is only given in Stiefel coordinates (Algorithm~\ref{alg:cgs}) as we do not have a closed-form expression for parallel transport in projection coordinates. 

We will rely on our embedding of  $\Graff(k, n)$ into $\Gr(k+1,n+1)$ via Stiefel or projection coordinates as given by Propositions~\ref{prop:graffquot} and \ref{prop:graffproj} respectively. We  then borrow the corresponding methods on the Grassmannian developed in \cite{AMS, EAS} in conjunction with Propositions~\ref{prop:basicstief} and \ref{prop:basicproj}.
\begin{algorithm}[h]
  \caption{Steepest descent in Stiefel coordinates}
  \label{alg:sds}
  \begin{algorithmic}
    \State Initialize $\mathbf{A}_0+b_{0} \in \Graff(k,n)$ in Stiefel coordinates $Y_0 \coloneqq Y_{\mathbf{A}_0+b_{0}} \in \mathbb{R}^{(n+1) \times (k+1)}$.
    \For {$i =0, 1, \dots$}
    \State set $G_{i} = f_{Y}(Y_{i}) - Y_{i}Y_{i}^\tp f_{Y}(Y_i)$; \Comment{gradient of $f$ at $Y_{i}$}
    \State compute $-G_i = U\Sigma V^\tp $; \Comment{condensed \textsc{svd}}
    \State minimize $f(Y(t))= f(Y_{i}V\cos(t\Sigma )V^\tp  + U\sin(t\Sigma )V^\tp ) \text{ over } t \in \mathbb{R}$; \Comment{exact line search}
    \State set $Y_{i+1} = Y(t_{\min})$;
    \EndFor 
  \end{algorithmic}
\end{algorithm}

\begin{algorithm}[h]
  \caption{Conjugate gradient in Stiefel coordinates}
  \label{alg:cgs}
  \begin{algorithmic}
    \State Initialize $\mathbf{A}_0+b_{0} \in \Graff(k,n)$ in Stiefel coordinates $Y_0 \coloneqq Y_{\mathbf{A}_0+b_{0}} \in \mathbb{R}^{(n+1) \times (k+1)}$.
    \State Set $G_{0} = f_{Y}(Y_{0}) - Y_{0}Y_{0}^\tp f_{Y}(Y_0)$ and $H_0 = -G_0$.
    \For {$i =0,1, \dots$}
    \State compute $H_i = U\Sigma V^\tp $; \Comment{condensed \textsc{svd}}
    \State minimize $f(Y(t))= f(Y_{i}V\cos(t\Sigma )V^\tp  + U\sin(t\Sigma )V^\tp ) \text{ over } t \in \mathbb{R}$; \Comment{exact line search}
    \State set $Y_{i+1} = Y(t_{\min})$;
    \State set $G_{i+1} = f_{Y}(Y_{i+1}) - Y_{i+1}Y_{i+1}^\tp f_{Y}(Y_{i+1})$; \Comment{gradient of $f$ at $Y_{i+1}$}
    \Procedure{Descent}{$Y_i,G_i,H_i$} \Comment{set new descent direction at $Y_{i+1}$}
    \State $\tau H_{i} = (-Y_{i}V\sin(t_{\min}\Sigma)+U\cos(t_{\min}\Sigma))\Sigma V^\tp $; \Comment{parallel transport of $H_i$}
    \State $\tau G_{i} = G_{i} - \bigl(Y_{i}V\sin(t_{\min}\Sigma) + U(I-\cos(t_{\min}\Sigma))\bigr)U^\tp G_{i}$; \Comment{parallel transport of $G_i$}
    \State $\gamma_i = \tr((G_{i+1}-\tau G_i)^\tp G_{i+1})/\tr(G_i^\tp G_i)$;
    \State $H_{i+1} = -G_{i+1} + \gamma_{i} \tau H_i$;
    \EndProcedure
    \State reset $H_{i+1} = -G_{i+1}$ if $i+1 \equiv 0 \mod (k+1)(n-k)$;
    \EndFor
  \end{algorithmic}
\end{algorithm}

\begin{algorithm}[h]
  \caption{Newton's method in Stiefel coordinates}
  \label{alg:nms}
  \begin{algorithmic}
    \State Initialize $\mathbf{A}_0+b_{0} \in \Graff(k,n)$ in Stiefel coordinates $Y_0 \coloneqq Y_{\mathbf{A}_0+b_{0}} \in \mathbb{R}^{(n+1) \times (k+1)}$.
    \For {$i =0, 1, \dots$}
    \State set $G_{i} = f_{Y}(Y_{i}) - Y_{i}Y_{i}^\tp f_{Y}(Y_i)$; \Comment{gradient of $f$ at $Y_{i}$}
    \State find $\Delta$ such that $Y_{i}^\tp \Delta = 0$ and $\nabla^2 f(\Delta) - \Delta(Y_{i}^\tp f_Y(Y_{i})) = -G$;
    \State compute $\Delta = U\Sigma V^\tp $; \Comment{condensed \textsc{svd}}
    \State $Y_{i+1} = Y_{i}V\cos(t\Sigma )V^\tp  + U\sin(t\Sigma )V^\tp $; \Comment{arbitrary step size $t$}
    \EndFor
  \end{algorithmic}
\end{algorithm}

\begin{algorithm}[h]
  \caption{Steepest descent in projection coordinates}
  \label{alg:sdp}
  \begin{algorithmic}
    \State Initialize $\mathbf{A}_0+b_{0} \in \Graff(k,n)$ in projection coordinates $P_0 \coloneqq P_{\mathbf{A}_0+b_{0}} \in \mathbb{R}^{(n+1) \times (n+1)}$.
    \For {$i =0,1, \dots$}
    \State set $\nabla f(P_i) = [P_i,[P_i,f_{P}(P_i)]]$;
    \State find $\Theta \in \mathbb{R}^{(n+1) \times (n+1)}$ and $t > 0$ so that
$P_i = \Theta^\tp \begin{bsmallmatrix} I_{k+1}&0 \\ 0&0 \end{bsmallmatrix} \Theta$ and $-t\nabla f(P_i) = \begin{bsmallmatrix} 0&Z \\ -Z^\tp &0 \end{bsmallmatrix} $;
    \State set $P_{i+1} =  \frac{1}{2} I_{n+1} + \Theta^{T} \begin{bsmallmatrix} \frac{1}{2}\cos(2\sqrt{ZZ^\tp })&-\sinc(2\sqrt{ZZ^\tp })Z \\ -Z^\tp \sinc(2\sqrt{ZZ^\tp })&-\frac{1}{2}\sin(2\sqrt{Z^\tp Z}) \end{bsmallmatrix}  \Theta$;
    \EndFor 
  \end{algorithmic}
\end{algorithm}

\begin{algorithm}[h]
  \caption{Newton's method in projection coordinates}
  \label{alg:nmp}
  \begin{algorithmic}
    \State Initialize $\mathbf{A}_0+b_{0} \in \Graff(k,n)$ in projection coordinates $P_0 \coloneqq P_{\mathbf{A}_0+b_{0}} \in \mathbb{R}^{(n+1) \times (n+1)}$.
    \For {$i =0,1, \dots$}
    \State find  $\Omega_i \in \mathsf{\Lambda}^2(\mathbb{R}^{n+1})$ such that
\[
[P_i,[P_i,  \nabla^2 f([P_i,[P_i,\Omega_i]])]] - [P_i,[\nabla f(P_i), [P_i,\Omega_i]]] = - [P_i,[P_i,\nabla f(P_i)]];
\]
    \State find $\Theta_i \in \SO(n+1)$ such that $P_i = \Theta_{i}^\tp  \begin{bsmallmatrix} I_{k+1}&0 \\ 0&0 \end{bsmallmatrix}  \Theta_i$; \Comment{QR factorization}
    \State compute  $\Theta_i(I-[P_i,[P_i,t\Omega_i]])\Theta_{i}^\tp  = Q_i R_i$; \Comment{QR factorization with positive diagonal in $R_i$}
    \State set $P_{i+1} = \Theta_{i}^\tp  Q_i \Theta_i P_i \Theta_{i}^\tp  Q_i^\tp \Theta_i$;
    \EndFor
  \end{algorithmic}
\end{algorithm}

There is one caveat: Algorithms~\ref{alg:sds}--\ref{alg:nmp} are formulated as \emph{infeasible methods}. If we start from a point in $\Graff(k,n)$, regarded as a subset of $\Gr(k+1,n+1)$, the next iterate along the geodesic may become \emph{infeasible}, i.e., fall outside $\Graff(k,n)$. By Theorem~\ref{thm:alg},
this will occurs with probability zero  but even if it does, the algorithms will still work fine as algorithms on $\Gr(k+1,n+1)$.

If desired, we may undertake a more careful \emph{prediction--correction} approach. Instead of having the points $Y_{i+1}$ (in Stiefel coordinates) or $P_{i+1}$ (in projection coordinates) be the next iterates, they will be `predictors' of the next iterates. We will then use Lemmas~\ref{lem:checkgraffstief} or \ref{lem:checkgraffproj} to check if $Y_{i+1}$ or $P_{i+1}$ are in $\Graff(k,n)$. In the unlikely scenario when they do fall outside $\Graff(k,n)$, e.g., if we have $Y_{i+1} = \begin{bsmallmatrix} A & b \\ 0 & \gamma \end{bsmallmatrix}$ where $A^\tp b \ne 0$ or $P_{i+1} = \begin{bsmallmatrix} S & d \\ d^\tp  & \gamma \end{bsmallmatrix}$ where $Sd \ne 0$, we will `correct' the iterates to feasible points $\widetilde{Y}_{i+1}$ or $\widetilde{P}_{i+1}$ by an appropriate reorthogonalization.

\section{Numerical experiments}\label{sec:numerical}

We will present various numerical experiments on two problems to illustrate the conjugate gradient and steepest descent algorithms in Section~\ref{sec:optim}. These problems are deliberately chosen to be non-trivial and yet have closed-form solutions --- so that we may check whether our algorithms have converged to the true solutions of these problems. We implemented  Algorithms~\ref{alg:sds} and \ref{alg:cgs} in \textsc{Matlab} and \textsc{Python} and used a combination of (i) Frobenius norm of the Riemannian gradient, (ii) distance between successive iterates, and (iii) number of iterations, for our stopping condition.

\subsection{Eigenvalue problem coupled with quadratic fractional programming}\label{sec:evpqfp}

Let $A \in \mathbb{R}^{n \times n}$ be symmetric, $b \in \mathbb{R}^{n}, $ and $c \in \mathbb{R}. $
We would like to solve
\begin{equation}\label{eq:evpqfp}
\begin{tabular}{rl}
minimize & $ \tr(X^\tp AX) + \dfrac{1}{1+ \lVert y\rVert^{2}}(y^\tp Ay + 2b^\tp y + c) $,\\[3ex]
subject to & $X^\tp X = I, \; X^\tp y = 0$,
\end{tabular}
\end{equation}
over all $X \in \mathbb{R}^{n \times k}$ and $y \in \mathbb{R}^n$. If we set $y = 0$ in \eqref{eq:evpqfp}, the resulting quadratic trace minimization problem with orthonormal constraints is essentially a \emph{symmetric eigenvalue problem}; if we set $X = 0$ in \eqref{eq:evpqfp}, the resulting nonconvex optimization problem is called \emph{quadratic fractional programming}.

By rearranging terms, \eqref{eq:evpqfp} transforms into a minimization problem over an affine Grassmannian,
\begin{equation}\label{eq:evpqfp2}
\min_{\mathbf{X}+y \in \Graff(k, n)} \tr \biggl( \begin{bmatrix} X&y/\sqrt{1 + \lVert y\rVert^{2}} \\ 0&1/\sqrt{1 + \lVert y\rVert^{2}} \end{bmatrix}^\tp \begin{bmatrix} A&b \\ b^\tp &c \end{bmatrix} \begin{bmatrix} X&y/\sqrt{1 + \lVert y\rVert^{2}} \\ 0&1/\sqrt{1 + \lVert y\rVert^{2}} \end{bmatrix} \biggr),
\end{equation}
which  shows that the problem \eqref{eq:evpqfp} is in fact coordinate independent, depending on  $X$ and $y$ only through the affine subspace  $\spn(X) + y = \mathbf{X} + y$. Formulated in this manner, we may determine a closed-form solution  via the eigenvalue decomposition of $\begin{bsmallmatrix} A&b \\ b^\tp &c \end{bsmallmatrix}$ --- the optimum value is the sum of the  $k+1$ smallest eigenvalues.

Figure~\ref{fig:convergence} shows convergence trajectories of steepest descent and conjugate gradient in Stiefel coordinates, i.e., Algorithms~\ref{alg:sds} and \ref{alg:cgs},  on  $\Graff(3,6)$ for the problem \eqref{eq:evpqfp2}.  $\Graff(3,6)$ is a $12$-dimensional manifold; we generate $A \in \mathbb{R}^{6 \times 6}$, $b \in \mathbb{R}^{6}$, $c \in \mathbb{R}$ randomly with $\mathcal{N}(0,1)$  entries, and likewise pick a random initial point in $\Graff(3,6)$. The gradient of $f(Y) \coloneqq \tr\bigl(Y^\tp \begin{bsmallmatrix} A&b \\ b^\tp &c \end{bsmallmatrix} Y\bigr)$ is given by $\nabla f(Y) = \begin{bsmallmatrix} A&b \\ b^\tp &c \end{bsmallmatrix}Y$. Both algorithms converge to the true solution but conjugate gradient converges twice as fast when measured by the number of iterations, taking around $20$ iterations for near-zero error reduction as opposed to steepest descent's $40$ iterations. The caveat is that each iteration of conjugate gradient is more involved and requires roughly twice the amount of time it takes for each iteration of steepest descent.
\begin{figure}
  \centering
    \includegraphics[scale=0.6]{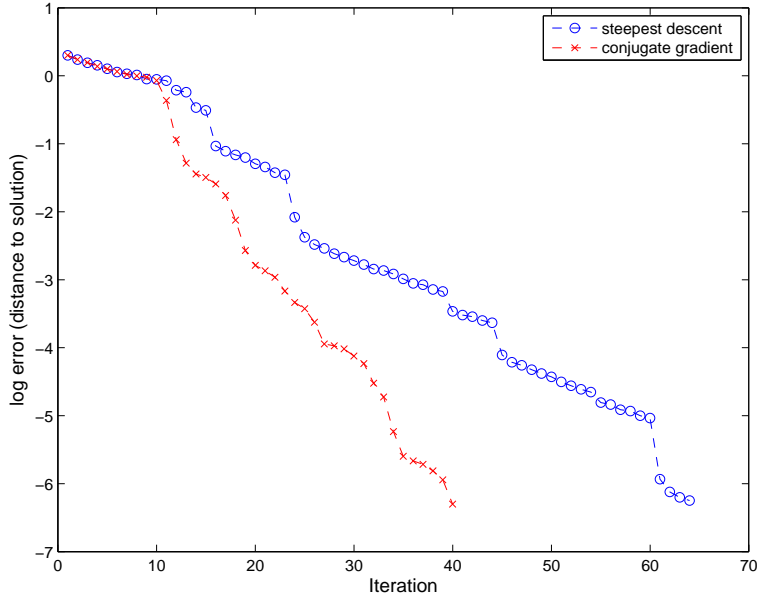}
\vspace{-3ex}
      \caption{Convergence trajectories of steepest descent and conjugate gradient for a quadratic optimization problem on the affine Grassmannian $\Graff(3,6)$.}
    \label{fig:convergence}
\end{figure}

We perform more extensive experiments by taking the average of $100$ instances of the problem \eqref{eq:evpqfp} for various values of $k$ and $n$  to generate  tables of timing and accuracy. Table~\ref{table:distevk} and \ref{table:distevn}  show the robustness of the algorithm with respect to different choices of $k$ and $n$.
\begin{table}[ht]
\renewcommand{\arraystretch}{1.3}
	\centering
	\begin{tabular}{r|rrrrrrrrrrr}
		$k$ & 10 & 21 & 32 & 43 & 54 & 65 & 76 & 87 & 98\\
		\hline\hline
		Steepest descent ($\times 10^{-6}$) & $0.61$ & $3.1$ & $1.5$ & $1.7$ & $2.9$ & $6.8$ & $1.2$ & $0.25$ & $0.1$ \\ 
		\hline
		Conjugate gradient ($\times 10^{-8}$) & $0.77$ & 1.5 & 1.9 & 2.4 & 2.3 & 2.9 & 3.1 & 3.5 & 3.3
	\end{tabular} 
\vspace{1ex}
	\caption{Accuracy (distance to true solution) of steepest descent and conjugate gradient for a quadratic optimization problem on  $\Graff(k,100)$.}
	\label{table:distevk}
\vspace{-4ex}
\end{table}
\begin{table}[ht]
\renewcommand{\arraystretch}{1.3}
	\centering
	\begin{tabular}{r|rrrrrrrrrrr}
		$n$ & 7 & 17 & 27 & 37 & 47 & 57 & 67 & 77 & 87\\
		\hline\hline
		Steepest descent ($\times 10^{-7}$) & $4.4$ & $4.8$ & $4.4$ & $4.7$ & $4.7$ & $4.7$ & $4.3$ & $4.7$ & $4.1$ \\ 
		\hline
		Conjugate gradient ($\times 10^{-6}$) & $0.83$ & $0.98$ & 1.0 & 1.3 & 1.2 & 1.3 & 1.5 & 1.6 & 1.5
	\end{tabular}
\vspace{1ex}
	\caption{Accuracy (distance to true solution) of steepest descent and conjugate gradient for a quadratic optimization problem on $\Graff(6,n)$.}
\label{table:distevn}
\vspace{-4ex}
\end{table}
\begin{table}[ht]
\renewcommand{\arraystretch}{1.3}
	\centering
	\begin{tabular}{r|rrrrrrrrrrr}
		 $k$ & 10 & 21 & 32 & 43 & 54 & 65 & 76 & 87 & 98\\
		\hline\hline
		Steepest descent & 0.6 & 0.89 & 1.4 & 1.4 & 1.8 & 1.9 & 2.0 & 2.0 & 1.3 \\ 
		\hline
		Conjugate gradient & 0.18 & 0.26 & 0.35 & 0.39 & 0.49 & 0.48 & 0.51 & 0.51 & 0.41 
	\end{tabular}
\vspace{1ex}
	\caption{Elapsed time (in seconds) of steepest descent and conjugate gradient for a quadratic optimization problem on $\Graff(k,100)$.}
\label{table:timeevk}
\vspace{-4ex}
\end{table}
\begin{table}[ht]
\renewcommand{\arraystretch}{1.3}
	\centering
	\begin{tabular}{r|rrrrrrrrrrr}
		 $n$ & 7 & 17 & 27 & 37 & 47 & 57 & 67 & 77 & 87\\
		\hline\hline
		Steepest descent & 0.67 & 0.96 & 0.94 & 1.1 & 1.2 & 1.3 & 1.4 & 1.4 & 1.5 \\ 
		\hline
		Conjugate gradient & 0.23 & 0.29 & 0.3 & 0.34 & 0.33 & 0.38 & 0.39 & 0.39 & 0.42
	\end{tabular}
\vspace{1ex}
	\caption{Elapsed time (in seconds) of steepest descent and conjugate gradient for a quadratic optimization problem on $\Graff(6,n)$}
	 \label{table:timeevn}
\end{table}

Table~\ref{table:timeevk} shows a modest initial increase followed by a  decrease in  elapsed time to convergence as $k$ increases --- a reflection of the intrinsic dimension of the problem as  $\dim\bigl(\Graff(k,100)\bigr) = (k+1)(100-k)$ first increases and then decreases. On the other hand, if we fix the dimension of ambient space, Table~\ref{table:timeevn} shows that the elapsed time increases with $k$. The results indicates that the elapsed time increases with the dimension of the affine Grassmannian.

\subsection{Fr\'echet mean and Karcher mean of affine subspaces}

Let  $d = d_{\Graff(k,n)}$, the geodesic distance on $\Graff(k,n)$ as defined in \eqref{eq:graffdist}. We would like to solve for the minimizer $\mathbf{X}+ y \in \Graff(k,n)$ in the sum-of-square-distances minimization problem:
\begin{equation}\label{eq:km1}
\min_{\mathbf{X}+y \in \Graff(k,n)} \sum\nolimits_{i=1}^{m} d^2(\mathbf{A}_i+b_i,\mathbf{X}+y),
\end{equation}
where $\mathbf{A}_i+b_i \in \Graff(k,n)$, $i=1,\dots,m.$ The Riemannian gradient \cite{K} of the objective function
\begin{equation}\label{eq:km2}
f_m(\mathbf{X}+y) = \sum\nolimits_{i=1}^{m} d^2(\mathbf{A}_i+b_i,\mathbf{X}+y)
\end{equation}
is given by
\[
\nabla f_m(\mathbf{X}+y) = \frac{1}{2} \sum\nolimits_{i=1}^{m}  \log_{\mathbf{X}+y} (\mathbf{A}_i+b_i),
\]
where $\log_{\mathbf{X}+y} (\mathbf{A}+b)$ denotes the derivative of the geodesic $\gamma(t)$ connecting $\mathbf{X}+y$ and $\mathbf{A}+b$ at $t=0$, with an explicit expression given by \eqref{eq:log}. 

The global minimizer of this problem is called the \emph{Fr\'echet mean} and a local minimizer is called a \emph{Karcher mean}. For the case $m =2$, they coincide and is given by the midpoint of the geodesic connecting $\mathbf{A}_1+b_1$ and $\mathbf{A}_2+b_2$, which has a closed-form expression given by \eqref{eq:geodesic} with $t =1/2$.

We will take the $\Graff(7,19)$, a $96$-dimensional manifold, as our specific example. Our objective function is $f_2(\mathbf{X}+y) = d^2(\mathbf{A}_1+b_1,\mathbf{X}+y) + d^2(\mathbf{A}_2+b_2,\mathbf{X}+y)$
and we set our initial point as one of the two affine subspaces.

\begin{figure}[h]
  \centering
    \includegraphics[scale=0.6]{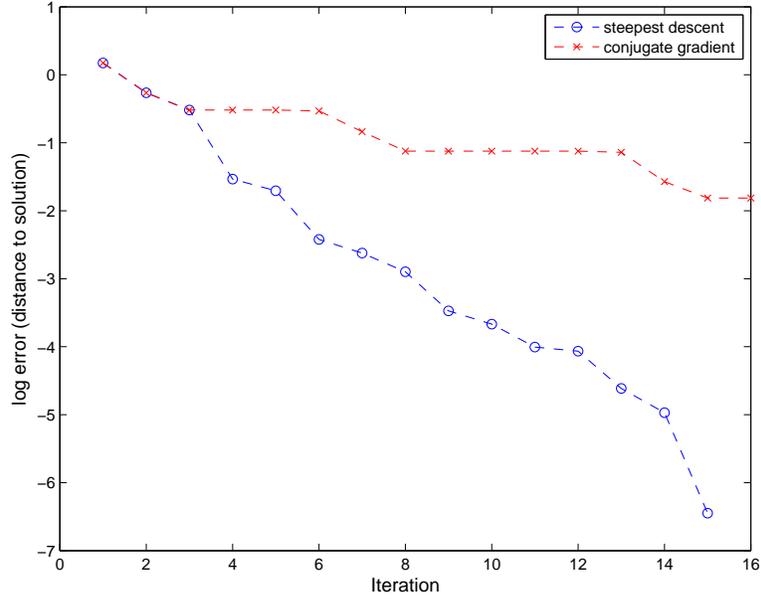}
\vspace{-3ex}
      \caption{Convergence trajectories of steepest descent and conjugate gradient  for Fr\'echet/Karcher mean  on the affine Grassmannian  $\Graff(7,19)$}
    \label{fig:convergence2}
\end{figure}

The result, depicted in Figure~\ref{fig:convergence2}, shows that steepest descent outperforms conjugate gradient in this specific example, unlike the example we considered in Section~\ref{sec:evpqfp}, which shows the opposite. So each algorithm serves a purpose for different types of problems.
In fact, when we  find the  Karcher mean of $m > 2$ affine subspaces by extending $f_m$ to the objective function in \eqref{eq:km1}, we see faster convergence (as measured by actual elapsed time) in conjugate gradient instead.
\begin{table}[ht]
\renewcommand{\arraystretch}{1.3}
	\centering
	\begin{tabular}{r|rrrrrrrrrrr}
		$k$ & 1 & 2 & 3 & 4 & 5 & 6 & 7 & 8 & 9\\
		\hline\hline
		Steepest descent ($\times 10^{-7}$) & $5.3$ & $5.1$ & $4.6$ & $4.8$ & $4.4$ & $4.9$ & $4.7$ & $4.6$ & $5.0$ \\ 
		\hline
		Conjugate gradient ($\times 10^{-1}$)  & 0.5 & 2.6 & 1.5 & 1.6 & 2.7 & 2.0 & 2.0 & 2.5 & 19.0 
	\end{tabular} 
\vspace{1ex}
\caption{Accuracy (distance to true solution) of steepest descent and conjugate gradient for Fr\'echet/Karcher mean  on  $\Graff(k,10)$.}
\label{table:distmeank}
\vspace{-3ex}
\end{table}
\begin{table}[ht]
\renewcommand{\arraystretch}{1.3}
	\centering
	\begin{tabular}{r|rrrrrrrrrrr}
		$n$ & 7 & 8 & 9 & 10 & 11 & 12 & 13 & 14 & 15\\
		\hline\hline
		Steepest descent ($\times 10^{-7}$) & $4.4$ & $4.8$ & $4.4$ & $4.7$ & $4.7$ & $4.7$ & $4.3$ & $4.7$ & $4.1$ \\ 
		\hline
		Conjugate gradient  ($\times 10^{-2}$)  & $0.36$ & 1.6 & 1.3 & 1.3 & 1.2 & 1.5 & 1.5 & 1.4 & 1.6
	\end{tabular}
\vspace{1ex}
	\caption{Accuracy (distance to true solution) of steepest descent and conjugate gradient for Fr\'echet/Karcher mean  on  $\Graff(6,n)$.}
\label{table:distmeann}
\vspace{-3ex}
\end{table}
\begin{table}[ht]
\renewcommand{\arraystretch}{1.3}
	\centering
	\begin{tabular}{r|rrrrrrrrrrr}
		 $k$ & 1 & 2 & 3 & 4 & 5 & 6 & 7 & 8 & 9\\
		\hline\hline
		Steepest descent ($\times 10^{-2}$) & 4.0 & 4.6 & 4.9 & 5.1 & 5.1 & 5.5 & 5.3 & 5.1 & 5.4 \\ 
		\hline
		Conjugate gradient ($\times 10^{-2}$) & 3.6 & 4.5 & 4.9 & 5.0 & 5.4 & 5.4 & 5.3 & 4.5 & 12.0
	\end{tabular}
\vspace{1ex}
	\caption{Elapsed time (in seconds) of steepest descent and conjugate gradient for Fr\'echet/Karcher mean on $\Graff(k,10)$.}
	 \label{table:timemeank}
\vspace{-3ex}
\end{table}
\begin{table}[ht]
\renewcommand{\arraystretch}{1.3}
	\centering
	\begin{tabular}{r|rrrrrrrrrrr}
		 $n$ & 7 & 8 & 9 & 10 & 11 & 12 & 13 & 14 & 15\\
		\hline\hline
		Steepest descent ($\times 10^{-1}$) & 3.1 & 3.0 & 3.5 & 3.3 & 3.7 & 4.1 & 3.8 & 4.1 & 4.3 \\ 
		\hline
		Conjugate gradient ($\times 10^{-1}$) & 17.0  & 2.1 & 2.8 & 3.1 & 3.4 & 3.8 & 3.9 & 3.6 & 3.6
	\end{tabular}
\vspace{1ex}
	\caption{Elapsed time (in seconds) of steepest descent and conjugate gradient for Fr\'echet/Karcher mean on $\Graff(6,n)$.}
	 \label{table:timemeann}
\end{table}

More extensive numerical experiments indicate that steepest descent and conjugate gradient are about equally fast for minimizing \eqref{eq:km2}, see Tables~\ref{table:timemeank}  and \ref{table:timemeann}, but that steepest decent is more accurate by orders of magnitude, see Tables~\ref{table:distmeank} and \ref{table:distmeann}. While these numerical experiments are intended for testing our algorithms, we would like to point out their potential application to \emph{model averaging}, i.e., aggregating affine subspaces estimated from different datasets.


\section{Conclusion}

We introduce the affine Grassmannian $\Graff(k,n)$, study its basic differential geometric properties, and develop several concrete systems of  coordinates --- three simple ones that are handy in proofs and two more sophisticated ones intended for computations; the latter two we called Stiefel and projection coordinates respectively. We show that when expressed in terms of Stiefel or projection coordinates, basic geometric objects  on $\Graff(k,n)$ may be readily represented as matrices and manipulated with standard routines in numerical linear algebra. With these in place, we ported the three standard Riemannian optimization algorithms  on the Grassmannian --- steepest descent, conjugate gradient, and Newton method --- to the affine Grassmannian. We demonstrated the efficacy of the first two algorithms through  extensive numerical experiments on two nontrivial problems with closed-form solutions, which allows us to ascertain the correctness of our results. The encouraging outcomes in these experiments provide a positive outlook towards further potential applications of our framework. Our hope is that numerical algorithms on the affine Grassmannian could become a mainstay in statistics and machine learning, where estimation problems may often be formulated as optimization problems on $\Graff(k,n)$.

\section*{Acknowledgment} We thank Pierre-Antoine Absil and Tingran Gao for very helpful discussions. 
This work  is generously supported by AFOSR FA9550-13-1-0133, DARPA D15AP00109, NSF IIS 1546413, DMS 1209136, DMS 1057064, National Key R\&D Program of China Grant 2018YFA0306702 and the NSFC Grant 11688101.  LHL's work is  supported by a DARPA Director's Fellowship and the Eckhardt Faculty Fund; KY's work is supported by the Hundred Talents Program of the Chinese Academy of Sciences and the Recruitment Program of the Global Experts of China.

\end{document}